\documentclass[11pt]{article}
\usepackage{hyperref}
\usepackage{times}  
\usepackage{mathpazo}
\usepackage{amssymb,amsmath,amsthm}
\usepackage{epsfig}

\newtheorem{theorem}{Theorem}[section]
\newtheorem{proposition}[theorem]{Proposition}
\newtheorem{definition}[theorem]{Definition}

\newtheorem{lemma}[theorem]{Lemma}

\newtheorem{remark}[theorem]{Remark}

\newcommand{\qedsymb}{\hfill{\rule{2mm}{2mm}}}
\renewenvironment{proof}[1][]{\begin{trivlist}
\item[\hspace{\labelsep}{\bf\noindent Proof#1:\/}] }{\qedsymb\end{trivlist}}

\newcommand{\PPA}{\mathsf{PPA}}
\newcommand{\TFNP}{\mathsf{TFNP}}
\newcommand{\PPAD}{\mathsf{PPAD}}
\newcommand{\PPP}{\mathsf{PPP}}
\newcommand{\PLS}{\mathsf{PLS}}
\newcommand{\CLS}{\mathsf{CLS}}
\newcommand{\EOPL}{\mathsf{EOPL}}

\newcommand{\NP}{\mathsf{NP}}

\newcommand{\FISC}{\textsc{Fair-IS-Cycle}}

\newcommand{\FSplitC}{\textsc{Fair-Split-Cycle}}
\newcommand{\FSplitP}{\textsc{Fair-Split-Path}}

\newcommand{\ConHalv}{\textsc{Con-Halving}}

\newcommand{\OTucker}{\textsc{Octahedral-Tucker}}
\newcommand{\LEAF}{\textsc{Leaf}}
\newcommand{\SchrijverP}{\textsc{Schrijver}}
\newcommand{\KneserP}{\textsc{Kneser}}

\newcommand{\alt}{\mathrm{alt}}

\newcommand{\eps}{\epsilon}
\renewcommand{\epsilon}{\varepsilon}

\begin{document}

\title{{\bf The Complexity of Finding Fair Independent Sets in Cycles}}

\author{
Ishay Haviv\thanks{School of Computer Science, The Academic College of Tel Aviv-Yaffo, Tel Aviv 61083, Israel.
Research supported in part by the Israel Science Foundation (grant No. 1218/20).
}
}

\date{}

\maketitle

\begin{abstract}
Let $G$ be a cycle graph and let $V_1,\ldots,V_m$ be a partition of its vertex set into $m$ sets.
An independent set $S$ of $G$ is said to {\em fairly represent} the partition if $|S \cap V_i| \geq \frac{1}{2} \cdot |V_i| -1$ for all $i \in [m]$.
It is known that for every cycle and every partition of its vertex set, there exists an independent set that fairly represents the partition (Aharoni et al., A Journey through Discrete Math.,~2017). We prove that the problem of finding such an independent set is $\PPA$-complete.
As an application, we show that the problem of finding a monochromatic edge in a Schrijver graph, given a succinct representation of a coloring that uses fewer colors than its chromatic number, is $\PPA$-complete as well.
The work is motivated by the computational aspects of the `cycle plus triangles' problem and of its extensions.

\end{abstract}


\section{Introduction}

In 1986, Du, Hsu, and Hwang~\cite{DHH93} conjectured that if a graph on $3n$ vertices is the disjoint union of a Hamilton cycle of length $3n$ and $n$ pairwise vertex-disjoint triangles then its independence number is $n$.
The conjecture has become known as the `cycle plus triangles' problem and has been strengthened by Erd{\"{o}}s~\cite{ErdosQuestions90}, who conjectured that such a graph is $3$-colorable.
Fleischner and Stiebitz~\cite{FleischnerS92} confirmed these conjectures in a strong form and proved, using an algebraic approach of Alon and Tarsi~\cite{AlonT92}, that such a graph is in fact $3$-choosable.
Their proof can also be viewed as an application of Alon's Combinatorial Nullstellensatz technique~\cite{AlonNull99}. Slightly later, an alternative elementary proof of the $3$-coloring result was given by Sachs~\cite{Sachs93}. However, none of these proofs supplies an efficient algorithm that given a graph on $3n$ vertices whose set of edges is the disjoint union of a Hamilton cycle and $n$ pairwise vertex-disjoint triangles finds a $3$-coloring of the graph or an independent set of size $n$. Questions on the computational aspects of the problem were posed in several works over the years (see, e.g.,~\cite{FleischnerS97,AlonChallenges02,Berczi017,AABCKLZ17}).

A natural extension of the problem of Du et al.~\cite{DHH93} is the following. Let $G$ be a cycle and let $V_1, \ldots ,V_m$ be a partition of its vertex set into $m$ sets.
We are interested in an independent set of $G$ that (almost) {\em fairly represents} the given partition, that is, an independent set $S$ of $G$ satisfying $|S \cap V_i| \geq \frac{1}{2} \cdot |V_i| -1$ for all $i \in [m] = \{1,\ldots,m\}$. The existence of such an independent set was proved in a work of Aharoni, Alon, Berger, Chudnovsky, Kotlar, Loebl, and Ziv~\cite{AABCKLZ17}. For the special case where all the sets $V_i$ are of size $3$, the proof technique of Aharoni et al.~\cite{AABCKLZ17} allowed them to show that there are {\em two} disjoint independent sets that fairly represent the partition, providing a new proof of a stronger form of the original conjecture of Du et al.~\cite{DHH93}. The results of~\cite{AABCKLZ17} were then extended in a work of Alishahi and Meunier~\cite{AlishahiM17}. A special case of one of their results is the following.

\begin{theorem}[\cite{AlishahiM17}]\label{thm:AlishahiM}
Let $G$ be a cycle on $n$ vertices and let $V_1,\ldots,V_m$ be a partition of its vertex set into $m$ sets.
Suppose that $n$ and $m$ have the same parity.
Then, there exist two disjoint independent sets $S_1$ and $S_2$ of $G$ covering all vertices but one from each $V_i$ such that for each $j \in \{1,2\}$, it holds that $|S_j \cap V_i| \geq \frac{1}{2} \cdot |V_i|-1$ for all $i \in [m]$.
\end{theorem}
\noindent
As shown by Black et al.~\cite{BCFPS20}, analogues of Theorem~\ref{thm:AlishahiM} for paths and for partitions into sets of odd sizes can also be proved using the approach of Aharoni et al.~\cite{AABCKLZ17}.

It is interesting to mention that although the statements of Theorem~\ref{thm:AlishahiM} and of its aforementioned variants are purely combinatorial, all of their known proofs are based on tools from topology. The use of topological methods in combinatorics was initiated by Lov\'asz~\cite{LovaszKneser} who applied the Borsuk-Ulam theorem~\cite{Borsuk33} from algebraic topology to prove a conjecture of Kneser~\cite{Kneser55} on the chromatic number of Kneser graphs. For integers $n \geq 2k$, the {\em Kneser graph} $K(n,k)$ is the graph whose vertices are all the $k$-subsets of $[n]$ where two sets are adjacent if they are disjoint. It was proved in~\cite{LovaszKneser} that the chromatic number of $K(n,k)$ is $n-2k+2$, a result that was strengthened and generalized by several researchers (see, e.g.,~\cite[Chapter~3]{MatousekBook}).
One such strengthening was obtained by Schrijver~\cite{SchrijverKneser78}, who studied the subgraph of $K(n,k)$ induced by the collection of all $k$-subsets of $[n]$ with no two consecutive elements modulo $n$. This graph is denoted by $S(n,k)$ and is commonly referred to as the {\em Schrijver graph}. It was proved in~\cite{SchrijverKneser78}, again by a topological argument, that the chromatic number of $S(n,k)$ is equal to that of $K(n,k)$.
As for Theorem~\ref{thm:AlishahiM}, the proof of Alishahi and Meunier~\cite{AlishahiM17} employs the Octahedral Tucker lemma that was applied by Matou{\v{s}}ek~\cite{Matousek04} in an alternative proof of Kneser's conjecture and can be viewed as a combinatorial formulation of the Borsuk-Ulam theorem (see also~\cite{Ziegler02}). The approach of Aharoni et al.~\cite{AABCKLZ17} and of Black et al.~\cite{BCFPS20}, however, is based on a direct application of the chromatic number of the Schrijver graph.
As before, these proofs are not constructive, in the sense that they do not suggest efficient algorithms for the corresponding search problems.
Understanding the computational complexity of these problems is the main motivation for the current work.

In 1994, Papadimitriou~\cite{Papa94} has initiated the study of the complexity of total search problems in view of the mathematical argument that lies at the existence proof of their solutions.
Let $\TFNP$ be the complexity class, defined in~\cite{MegiddoP91}, of the total search problems in $\NP$, that is, the class of search problems in which a solution is guaranteed to exist and can be verified in polynomial running-time.
Papadimitriou has introduced in~\cite{Papa94} several subclasses of $\TFNP$, each of which consists of the total search problems that can be reduced to a problem that represents some mathematical argument.
For example, the class $\PPA$ (Polynomial Parity Argument) corresponds to the simple fact that every graph with maximum degree $2$ that has a vertex of degree $1$ must have another degree $1$ vertex. Hence, $\PPA$ is defined as the class of all problems in $\TFNP$ that can be efficiently reduced to the $\LEAF$ problem, in which given a succinct representation of a graph with maximum degree $2$ and given a vertex of degree $1$ in the graph, the goal is to find another such vertex.
The class $\PPAD$ (Polynomial Parity Argument in Directed graphs) is defined similarly with respect to directed graphs.
Another complexity class defined in~\cite{Papa94} is $\PPP$ (Polynomial Pigeonhole Principle) whose underlying mathematical argument is the pigeonhole principle. Additional examples of complexity classes defined in this way are $\PLS$ (Polynomial Local Search)~\cite{JohnsonPY88}, $\CLS$ (Continuous Local Search)~\cite{DaskalakisP11}, and $\EOPL$ (End of Potential Line)~\cite{FearnleyGMS19}.

The complexity class $\PPAD$ is known to perfectly capture the complexity of many important search problems.
Notable examples of $\PPAD$-complete problems are those associated with Sperner's lemma~\cite{Papa94,ChenD09}, the Nash Equilibrium theorem~\cite{ChenDT09,DaskalakisGP09}, the Envy-Free Cake Cutting theorem~\cite{DengQS12}, and the Hairy Ball theorem~\cite{GoldbergH19}.
For $\PPA$, the undirected analogue of $\PPAD$, until recently no `natural' complete problems were known, where by `natural' we mean that their definitions do not involve circuits and Turing machines. In the last few years, the situation was changed following a breakthrough result of Filos-Ratsikas and Goldberg~\cite{FG18,FG19}, who proved that the Consensus Halving problem with an inverse-polynomial precision parameter is $\PPA$-complete (see also~\cite{FHSZ20}) and used it to derive the $\PPA$-completeness of the classical Splitting Necklace with two thieves and Discrete Sandwich problems. This was obtained building on the $\PPA$-hardness, proved by Aisenberg, Bonet, and Buss~\cite{AisenbergBB20}, of the search problem associated with Tucker's lemma. The variant of the problem that corresponds to the Octahedral Tucker lemma was suggested for study by P{\'{a}}lv{\"{o}}lgyi~\cite{Palvolgyi09} and proved to be $\PPA$-complete by Deng, Feng, and Kulkarni~\cite{DengFK17}.
The $\PPA$-completeness of the Consensus Halving problem was improved to a constant precision parameter in a recent work of Deligkas, Fearnley, Hollender, and Melissourgos~\cite{DeligkasFHM22}.
Additional examples of $\PPA$-complete problems can be found, for instance, in the works of Belovs et al.~\cite{BelovsIQSY17}, Schnider~\cite{Schnider21a}, and Deligkas et al.~\cite{DeligkasFM22}.

\subsection{Our Contribution}

The present work initiates the study of the complexity of finding independent sets that fairly represent a given partition of the vertex set of a cycle.
It is motivated by the computational aspects of combinatorial existence statements, such as the `cycle plus triangles' conjecture of Du et al.~\cite{DHH93} proved by Fleischner and Stiebitz~\cite{FleischnerS92} and its extensions by Aharoni et al.~\cite{AABCKLZ17}, Alishahi and Meunier~\cite{AlishahiM17}, and Black et al.~\cite{BCFPS20}.
As mentioned before, the challenge of understanding the complexity of the corresponding search problems was explicitly raised by several authors, including Fleischner and Stiebitz~\cite{FleischnerS97}, Alon~\cite{AlonChallenges02}, and Aharoni et al.~\cite{AABCKLZ17}. In this work we demonstrate that this research avenue may illuminate interesting connections between this family of problems and the complexity class $\PPA$. As an application, we determine the complexity of finding a monochromatic edge in Schrijver graphs colored by fewer colors than the chromatic number.

We start by introducing the Fair Independent Set in Cycle Problem, which we denote by $\FISC$ and define as follows.
\begin{definition}[Fair Independent Set in Cycle Problem]\label{def:FISC}
In the $\FISC$ problem, the input consists of a cycle $G$ and a partition $V_1, \ldots ,V_m$ of its vertex set into $m$ sets. The goal is to find an independent set $S$ of $G$ satisfying $|S \cap V_i| \geq \frac{1}{2} \cdot |V_i|-1$ for all $i \in [m]$.
\end{definition}
\noindent
The existence of a solution to every input of $\FISC$ is guaranteed by a result of Aharoni et al.~\cite[Theorem~1.8]{AABCKLZ17}.
Since such a solution can be verified in polynomial running-time, the total search problem $\FISC$ lies in the complexity class $\TFNP$.
We prove that the class $\PPA$ captures the complexity of the problem.
\begin{theorem}\label{thm:FISC}
The $\FISC$ problem is $\PPA$-complete.
\end{theorem}
\noindent
In view of the `cycle plus triangles' problem of Du et al.~\cite{DHH93}, it would be interesting to understand the complexity of the $\FISC$ problem restricted to partitions into sets of size $3$. While Theorem~\ref{thm:FISC} immediately implies that this restricted problem lies in $\PPA$, the question of determining its precise complexity remains open.

We proceed by considering the search problem associated with Theorem~\ref{thm:AlishahiM}.
In the Fair Splitting of Cycle Problem, denoted $\FSplitC$, we are given a cycle and a partition of its vertex set and the goal is to find {\em two} disjoint independent sets that fairly represent the partition and cover all vertices but one from every part of the partition. We define below an approximate version of this problem, in which the fairness requirement is replaced with the relaxed notion of $\eps$-fairness, where the independent sets should include at least $\frac{1}{2}-\eps$ fraction of the vertices of every part.

\begin{definition}[Approximate Fair Splitting of Cycle Problem]\label{def:FSplitC}
In the $\eps\textsc{-}\FSplitC$ problem with parameter $\eps \geq 0$, the input consists of a cycle $G$ on $n$ vertices and a partition $V_1, \ldots ,V_m$ of its vertex set into $m$ sets, such that $n$ and $m$ have the same parity. The goal is to find two disjoint independent sets $S_1$ and $S_2$ of $G$ covering all vertices but one from each $V_i$ such that for each $j \in \{1,2\}$, it holds that $|S_j \cap V_i| \geq (\frac{1}{2}-\eps) \cdot |V_i|-1$ for all $i \in [m]$. For $\eps=0$, the problem is denoted by $\FSplitC$.
\end{definition}
\noindent
The existence of a solution to every input of $\eps\textsc{-}\FSplitC$, already for $\eps=0$, is guaranteed by Theorem~\ref{thm:AlishahiM} proved in~\cite{AlishahiM17}.
For $\eps=0$, it can be seen that $\FSplitC$ is at least as hard as $\FISC$ (see Lemma~\ref{lem:FISC2FSplitC}). Yet, it turns out that $\FSplitC$ lies in $\PPA$ and is thus also $\PPA$-complete.

\begin{theorem}\label{thm:FSplitC}
The $\FSplitC$ problem is $\PPA$-complete.
\end{theorem}
\noindent
In fact, using the recent work~\cite{DeligkasFHM22}, we also obtain the following $\PPA$-completeness result for the approximate version of the problem (see Remark~\ref{remark:0.1}).

\begin{theorem}\label{thm:FSplitC_eps}
There exists a constant $\eps>0$ for which the $\eps\textsf{-}\FSplitC$ problem is $\PPA$-complete.
\end{theorem}

We finally consider the complexity of the $\SchrijverP$ problem. In this problem we are given a succinct representation of a coloring of the Schrijver graph $S(n,k)$ with $n-2k+1$ colors, which is one less than its chromatic number~\cite{SchrijverKneser78}, and the goal is to find a monochromatic edge (see Definition~\ref{def:Schrijver}).
The study of the $\SchrijverP$ problem is motivated by a question raised by Deng et al.~\cite{DengFK17} regarding the complexity of the analogue $\KneserP$ problem for Kneser graphs.
Note that the latter is not harder than the $\SchrijverP$ problem, because $S(n,k)$ is a subgraph of $K(n,k)$ with the same chromatic number.
As an application of our Theorem~\ref{thm:FISC}, we prove the following.

\begin{theorem}\label{thm:SchrijverP}
The $\SchrijverP$ problem is $\PPA$-complete.
\end{theorem}

It would be interesting to determine the computational complexity of the $\KneserP$ problem and to decide whether it is $\PPA$-complete, as suggested in~\cite{DengFK17}. It would also be interesting to prove unconditional lower bounds on the query complexity of algorithms for the $\KneserP$ and $\SchrijverP$ problems in the black-box input model, where the input is given as an oracle access.
We note that the study of the $\KneserP$ problem is motivated by its connections to a resource allocation problem called Agreeable Set, that was introduced by Manurangsi and Suksompong~\cite{ManurangsiS19} and further studied in~\cite{GoldbergHIMS20,Haviv22a}.
From an algorithmic point of view, it was shown in the recent works~\cite{Haviv22a,Haviv22b} that there exist randomized algorithms for the $\KneserP$ and $\SchrijverP$ problems with running time $n^{O(1)} \cdot k^{O(k)}$ on graphs $K(n,k)$ and $S(n,k)$ respectively, hence these problems are fixed-parameter tractable with respect to the parameter $k$. It would be nice to further explore algorithms for these problems as well as for the other problems studied in the current work.

\subsection{Overview of Proofs}

To obtain our results we present a chain of reductions, as described in Figure~\ref{fig:reductions}.
Our starting point is the Consensus Halving problem with precision parameter $\eps$, in which given a collection of $m$ probability measures on the interval $[0,1]$ the goal is to partition the interval into two pieces using relatively few cuts, so that each of the measures has the same mass on the two pieces up to an error of $\eps$ (see Definition~\ref{def:ConHal}).
It is known that every instance of this problem has a solution with at most $m$ cuts even for $\eps=0$~\cite{SimmonsS03} (see also~\cite{GW85,AW86}) and that the problem of finding such a solution is $\PPA$-hard for some constant $\eps >0$~\cite{DeligkasFHM22}.
\begin{figure}[t!]
\begin{center}
\includegraphics[width=3.3in]{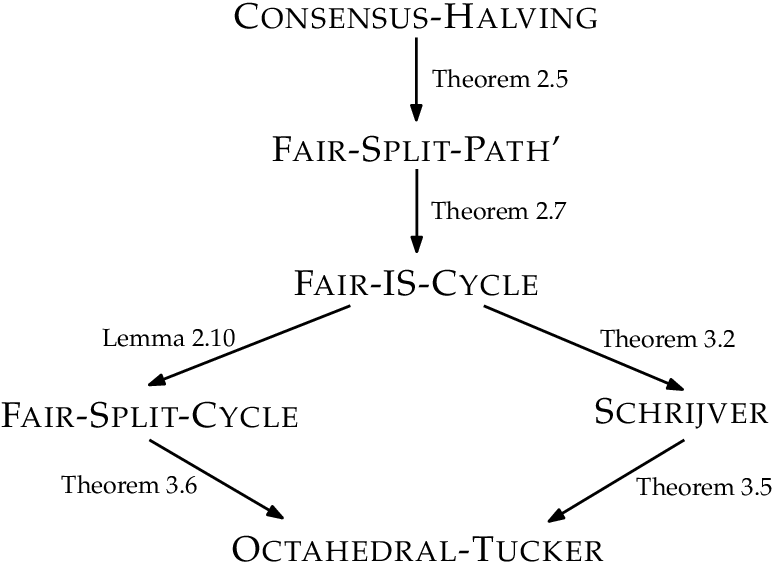}
\end{center}
\caption{Chain of reductions.}
\label{fig:reductions}
\end{figure}

In Section~\ref{sec:FISC}, we reduce the Consensus Halving problem to an intermediate variant of the $\eps\textsf{-}\FSplitC$ problem, which we call $\eps\textsf{-}\FSplitP'$ (see Definition~\ref{def:FSplitP'}). Then, we use this reduction to obtain our hardness results for the $\FISC$ and $\FSplitC$ problems.
The reduction borrows a discretization argument that was used in~\cite{FG18} to reduce the Consensus Halving problem to the Splitting Necklace problem with two thieves.
This argument enables us to transform a Consensus Halving instance into a path and a partition of its vertex set, for which the goal is to partition the path using relatively few cuts into two parts, each of which contains roughly half of the vertices of every set in the partition. In order to relate this problem to independent sets that fairly represent the partition, we need an additional simple trick. Between every two consecutive vertices of the path we add a new vertex and put all the new vertices in a new set added to the partition of the vertex set.
We then argue, roughly speaking, that two disjoint independent sets in the obtained path, which fairly represent the partition and cover almost all of the vertices, can be used to obtain a solution to the original instance. The high-level idea is that those few vertices that are uncovered by the two independent sets can be viewed as cuts, and every path between two such vertices alternates between the two given independent sets. By construction, it means that only one of the two independent sets contains in such a path original vertices (that is, vertices that were not added in the last phase of the reduction), hence every such path can be naturally assigned to one of the two pieces required by the Consensus Halving problem.
Combining our reduction with the known hardness results of Consensus Halving, we derive the $\PPA$-hardness of $\FISC$ and of $\eps\textsf{-}\FSplitC$ for a constant $\eps >0$, as needed for Theorems~\ref{thm:FISC},~\ref{thm:FSplitC}, and~\ref{thm:FSplitC_eps}.

In Section~\ref{sec:Schrijver}, we introduce and study the $\SchrijverP$ problem.
We reduce the $\FISC$ problem to the $\SchrijverP$ problem, implying that the latter is $\PPA$-hard.
The reduction follows an argument of Aharoni et al.~\cite{AABCKLZ17} who used the chromatic number of the Schrijver graph~\cite{SchrijverKneser78} to prove the existence of the independent set required in $\FISC$.
Finally, employing arguments of Meunier~\cite{Meunier11} and Alishahi and Meunier~\cite{AlishahiM17}, we reduce the $\SchrijverP$ and $\FSplitC$ problems to the search problem associated with the Octahedral Tucker lemma (see Definition~\ref{def:OTucker}). Since it is known, already from~\cite{Papa94}, that this problem lies in $\PPA$, we get that $\FISC$, $\FSplitC$, and $\SchrijverP$ are all members of $\PPA$, completing the proofs of Theorems~\ref{thm:FISC}, \ref{thm:FSplitC}, \ref{thm:FSplitC_eps}, and~\ref{thm:SchrijverP}.

We remark that one could consider analogues of the $\FISC$ and $\FSplitC$ problems for paths rather than for cycles and obtain similar results.
We have chosen to focus here on the cycle setting, motivated by the computational aspects of the `cycle plus triangles' problem~\cite{DHH93,ErdosQuestions90,FleischnerS92}.

\section{Fair Independent Sets in Cycles}\label{sec:FISC}

In this section we prove our hardness results for the $\FISC$ and $\FSplitC$ problems.
We first recall the definition of the Consensus Halving problem and state its hardness result from~\cite{DeligkasFHM22}.
Then, we present an efficient reduction from this problem to an intermediate problem, which is used to obtain the hardness results of Theorems~\ref{thm:FISC},~\ref{thm:FSplitC}, and~\ref{thm:FSplitC_eps}.

\subsection{Consensus Halving}

Consider the following variant of the Consensus Halving problem, denoted $\ConHalv$.
\begin{definition}[Consensus Halving Problem]\label{def:ConHal}
In the $\eps\textsf{-}\ConHalv(m,\ell)$ problem with precision parameter $\eps = \eps(m)$, the input consists of $m$ probability measures $\mu_1, \ldots, \mu_m$ on the interval $I=[0,1]$, given by their density functions.
The goal is to partition the interval $I$ using at most $\ell$ cuts into two (not necessarily connected) pieces $I^+$ and $I^-$, so that for every $i \in [m]$ it holds that $| \mu_i(I^+) - \mu_i(I^-)| \leq \eps$.
\end{definition}
\noindent
For $\ell \geq m$, every input of $\eps\textsf{-}\ConHalv(m,\ell)$ has a solution even for $\eps=0$~\cite{SimmonsS03}.
We will rely on the hardness result of $\ConHalv$ stated below.
Here, a function on an interval is said to be {\em piecewise constant} if its domain can be partitioned into a finite set of intervals such that the function is constant on each of them. We refer to the intervals of the partition on which the function is nonzero as the {\em blocks} of the function.
Note that a piecewise constant function can be explicitly represented by the endpoints and the values of its blocks.

\begin{theorem}[{\cite[Theorem~1.3]{DeligkasFHM22}}]\label{thm:CH_PPA_hard}
There exists a constant $\eps >0$ such that for every constant $c \geq 0$, the $\eps\textsf{-}\ConHalv(m,m+c)$ problem, restricted to inputs with piecewise constant density functions with at most $3$ blocks, is $\PPA$-hard.
\end{theorem}

\begin{remark}
We note that, as explained in~\cite{FHSZ20}, the constant $c$ given in Theorem~\ref{thm:CH_PPA_hard} can be replaced by $m^{1-\alpha}$ for any constant $\alpha >0$. This stronger hardness, however, is not required to obtain our results.
We also note that our results do not rely on the fact that the hardness given in Theorem~\ref{thm:CH_PPA_hard} holds for instances with density functions with at most $3$ blocks, as proved in~\cite{DeligkasFHM22}, rather than polynomially many blocks.
\end{remark}

\subsection{The Main Reduction}

To obtain our hardness results for the $\FISC$ and $\FSplitC$ problems, we consider the following intermediate problem.
\begin{definition}\label{def:FSplitP'}
In the $\eps\textsc{-}\FSplitP'$ problem with parameter $\eps \geq 0$, the input consists of a path $G$ and a partition $V_1, \ldots ,V_m$ of its vertex set into $m$ sets such that $|V_i|$ is odd for all $i \in [m]$. The goal is to find two disjoint independent sets $S_1$ and $S_2$ of $G$ covering all but at most $m$ of the vertices of $G$ such that
\[ |S_1 \cap V_i| \in \Big [  (\tfrac{1}{2}-\eps  ) \cdot |V_i|-1, (\tfrac{1}{2}+\eps ) \cdot |V_i| \Big ]\] for all $i \in [m]$.
When $\eps=0$, the problem is denoted by $\FSplitP'$.
\end{definition}
\noindent
Note that the $\eps\textsc{-}\FSplitP'$ problem differs from the $\eps\textsc{-}\FSplitC$ problem (see Definition~\ref{def:FSplitC}) in the following respects: (a) The input graph is a path rather than a cycle, (b) an $\eps$-fairness property is required only for the independent set $S_1$ rather than for both $S_1$ and $S_2$, (c) there is no requirement regarding the sets $V_i$ to which the vertices that are uncovered by $S_1$ and $S_2$ belong, and (d) the sets $V_i$ are required to be of odd sizes.
Yet, every instance of the $\eps\textsc{-}\FSplitP'$ problem has a solution already for $\eps=0$, as follows from Theorem~\ref{thm:AlishahiM} applied to the cycle obtained by connecting the endpoints of the given path by an edge.

We turn to prove the following.
\begin{theorem}\label{thm:MainReduction}
Let $p$ be a polynomial and suppose that $\eps = \eps(m)$ is bounded from below by some inverse-polynomial in $m$.
Then, for any constant $\alpha \in [0,1)$, the $\eps\textsf{-}\ConHalv(m,m+1)$ problem, restricted to inputs with piecewise constant density functions with at most $p(m)$ blocks, is polynomial-time reducible to the $\frac{\alpha \cdot \eps}{2}\textsc{-}\FSplitP'$ problem.
\end{theorem}

\begin{proof}
Consider an instance of $\eps\textsc{-}\ConHalv(m,m+1)$ consisting of $m$ probability measures $\mu_1, \ldots, \mu_m$ on the interval $I=[0,1]$, given by their piecewise constant density functions $g_1, \ldots, g_m$, each of which has at most $p(m)$ blocks.
Fix any constant $\alpha \in [0,1)$.
The reduction constructs an instance of $\frac{\alpha \cdot \eps}{2}\textsc{-}\FSplitP'$, namely, a path $G$ and a partition $V_1, \ldots, V_{m+1}$ of its vertex set into $m+1$ sets of odd sizes.

We start with a high-level description of the reduction.
First, borrowing a discretization argument of~\cite{FG18}, the reduction associates with every density function $g_i$ a collection $V_i$ of vertices located in the (at most $p(m)$) intervals on which $g_i$ is nonzero.
To do so, we partition every block of $g_i$ into sub-intervals such that the measure of $\mu_i$ on each of them is $\delta$, where $\delta>0$ is some small parameter (assuming, for now, that the measure of $\mu_i$ on every block is an integer multiple of $\delta$). At the middle of every such sub-interval we locate a vertex and put it in $V_i$. Then, we construct a path $G$ that alternates between the vertices of $V_1 \cup \cdots \cup V_m$ ordered according to their locations in $I$ and additional vertices which we put in another set $V_{m+1}$. We also take care of the requirement that each $|V_i|$ is odd.

The intuitive idea behind this reduction is the following. Suppose that we are given a solution to the constructed instance, i.e., two disjoint independent sets $S_1$ and $S_2$ of the path $G$ covering all but $m+1$ of the vertices such that $S_1$ contains roughly half of the vertices of $V_i$ for each $i \in [m+1]$.
Observe that by removing from $G$ the $m+1$ vertices that {\em do not} belong to $S_1 \cup S_2$, we essentially get a partition of the vertices of $S_1 \cup S_2$ into $m+2$ paths. Since $S_1$ and $S_2$ are independent sets in $G$, it follows that each such path alternates between $S_1$ and $S_2$. However, recalling that $G$ alternates between $V_1 \cup \cdots \cup V_m$ and $V_{m+1}$, it follows that ignoring the vertices of $V_{m+1}$, each such path contains either only vertices of $S_1$ or only vertices of $S_2$. Now, one can view the $m+1$ locations of the vertices that do not belong to $S_1 \cup S_2$ as cuts in the interval $I$ which partition it into $m+2$ sub-intervals, each of which includes vertices from either $S_1$ or $S_2$ (again, ignoring the vertices of $V_{m+1}$). Let $I^+$ and $I^-$ be the pieces of $I$ obtained from the sub-intervals that correspond to $S_1$ and $S_2$ respectively. Since the number of vertices from $V_i$ in every path is approximately proportional to the measure of $\mu_i$ in the corresponding sub-interval, it can be shown that the probability measure of $\mu_i$ on $I^+$ is approximately $\frac{1}{2}$. This yields that the probability measure $\mu_i$ is approximately equal on the pieces $I^+$ and $I^-$, as needed for the $\ConHalv(m,m+1)$ problem.

We turn to the formal description of the reduction. For an illustration, see Figure~\ref{fig:density}.
Define $\delta = \frac{(1-\alpha) \cdot \eps}{2 \cdot (2p(m)+m+3)}$. The reduction acts as follows.
\begin{enumerate}
  \item For every $i \in [m]$, do the following:
  \begin{itemize}
    \item We are given a partition of the interval $I$ into intervals such that on at most $p(m)$ of them the function $g_i$ is equal to a nonzero value and is zero everywhere else. For every such interval, let $\gamma$ denote the volume of $g_i$ on it, and divide it into $ \lceil \gamma/\delta \rceil$ sub-intervals of volume $\delta$ each, possibly besides the last one whose volume might be smaller. We refer to a sub-interval of volume smaller than $\delta$ as an {\em imperfect} sub-interval. The number of imperfect sub-intervals associated with $g_i$ is clearly at most $p(m)$. At the middle point of every sub-interval of $g_i$, locate a vertex and put it in the set $V_i$.
    \item If the number of vertices in $V_i$ is even, then add another vertex to $V_i$ and locate it arbitrarily in $I$.
    \item Note that, by $\mu_i(I)=1$, we have
    \begin{eqnarray}\label{eq:|V_i|}
    |V_i| \cdot \delta \in [1,1+(p(m)+1) \cdot \delta].
    \end{eqnarray}
  \end{itemize}
  \item Consider the path on the vertices of $V_1 \cup \cdots \cup V_m$ ordered according to their locations in the interval $I$, breaking ties arbitrarily.
  \item Add a new vertex before every vertex in this path, locate it at the middle of the sub-interval between its two adjacent vertices (where the first new vertex is located at $0$), and put these new vertices in the set $V_{m+1}$. If the number of vertices in $V_{m+1}$ is even then add one more vertex to the end of the path, locate it at $1$, and put it in $V_{m+1}$ as well.
      Denote by $G$ the obtained path, and note that $G$ alternates between $V_1 \cup \cdots \cup V_{m}$ and $V_{m+1}$.
  \item The output of the reduction is the path $G$ and the partition $V_1 , \ldots , V_{m+1}$ of its vertex set $V$ into $m+1$ sets. By construction, $|V_i|$ is odd for every $i \in [m+1]$.
\end{enumerate}
It is easy to verify that the reduction can be implemented in polynomial running-time. Indeed, every density function $g_i$ is piecewise constant with at most $p(m)$ blocks, hence for every $i \in [m]$ the number of vertices that the reduction defines for $V_i$ is at most $1/\delta+p(m)+1$, and the latter is polynomial in the input size because of the definition of $\delta$ and the fact that $\eps$ is at least inverse-polynomial in $m$.
The additional set $V_{m+1}$ doubles the number of vertices, possibly with one extra vertex, preserving the construction polynomial in the input size.

\begin{figure}[h]
\begin{center}
\includegraphics[width=3.3in]{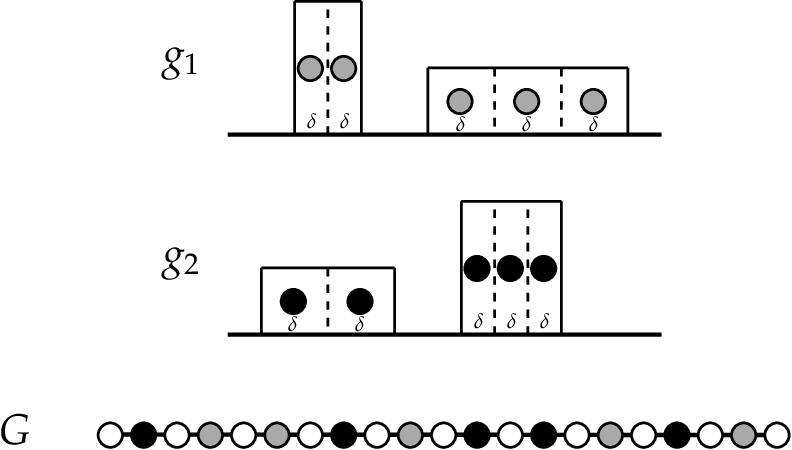}
\end{center}
\caption{An illustration of the reduction for $m=2$ (Theorem~\ref{thm:MainReduction}). Given the density functions $g_1$ and $g_2$, the reduction produces a path $G$ and a partition of its vertex set into three sets of odd sizes: $V_1$ (gray), $V_2$ (black), and $V_3$ (white). The path $G$ alternates between $V_1 \cup V_2$ and $V_3$.}
\label{fig:density}
\end{figure}

We turn to prove the correctness of the reduction, that is, that a solution to the constructed instance of $\frac{\alpha \cdot \eps}{2}\textsc{-}\FSplitP'$ can be used to efficiently compute a solution to the original instance of $\eps\textsf{-}\ConHalv(m,m+1)$.
Suppose we are given a solution to $\frac{\alpha \cdot \eps}{2}\textsc{-}\FSplitP'$ for the path $G$ and the partition $V_1 , \ldots , V_{m+1}$ of its vertex set $V$. Such a solution consists of two disjoint independent sets $S_1$ and $S_2$ of $G$ covering all but at most $m+1$ of the vertices of $G$ such that
\begin{eqnarray}\label{eq:SandV_i}
|S_1 \cap V_i| \in \Big [ (\tfrac{1}{2}-\tfrac{\alpha \cdot \eps}{2} ) \cdot |V_i|-1, (\tfrac{1}{2}+\tfrac{\alpha \cdot \eps}{2} ) \cdot |V_i| \Big ]
\end{eqnarray}
for all $i \in [m+1]$. Let $S_3 = V \setminus (S_1 \cup S_2)$. It can be assumed that $|S_3|=m+1$ (otherwise, remove some arbitrary vertices from $S_2$).
Denote the vertices of $S_3$ by $u_1, \ldots ,u_{m+1}$ ordered according to their order in $G$.
Let $P_1, \ldots, P_{m+2}$ be the $m+2$ paths obtained from $G$ by removing the vertices of $S_3$ (where some of the paths might be empty).
Since $S_1$ and $S_2$ are independent sets, every path $P_j$ alternates between $S_1$ and $S_2$. By our construction, this implies that in every path $P_j$ either the vertices of $S_1$ are from $V \setminus V_{m+1}$ and those of $S_2$ are from $V_{m+1}$, or the vertices of $S_2$ are from $V \setminus V_{m+1}$ and those of $S_1$ are from $V_{m+1}$.
We define $b_j = 1$ in the former case and $b_j=2$ in the latter.
Thus, for every $i \in [m]$, the number of vertices of $V_i$ that appear in the paths $P_j$ with $b_j=1$ is precisely $|S_1 \cap V_i|$.

Now, let $\beta_1, \ldots, \beta_{m+1} \in I$ be the locations of the vertices $u_1, \ldots, u_{m+1}$ in the interval $I$ as defined by the reduction. We interpret these locations as $m+1$ cuts of the interval $I$. Set $\beta_0 = 0$ and $\beta_{m+2}=1$, and for every $j \in [m+2]$, let $I_j$ denote the interval $[\beta_{j-1}, \beta_{j}]$. Consider the partition of $I$ into two pieces $I^+$ and $I^-$, where $I^+$ includes all the parts $I_j$ with $b_j=1$ and $I^-$ includes all the parts $I_j$ with $b_j=2$. We claim that this partition, which is obtained using $m+1$ cuts in $I$, forms a valid solution to the original instance of $\eps\textsf{-}\ConHalv(m,m+1)$. To this end, we show that for every $i \in [m]$ it holds that $|\mu_i(I^+) -\frac{1}{2}| \leq \frac{\eps}{2}$, which is equivalent to $|\mu_i(I^+) - \mu_i(I^-)| \leq \eps$.

Fix some $i \in [m]$. We turn to estimate the quantity $\mu_i(I^+)$, i.e., the total measure of $\mu_i$ on the intervals $I_j$ with $b_j=1$.
By our construction, every vertex of $V_i$ corresponds to a sub-interval whose measure by $\mu_i$ is $\delta$ (except for at most $p(m)+1$ of them). Since the intervals of $I^+$ correspond to the paths $P_j$ whose vertices in $V \setminus V_{m+1}$ are precisely the vertices of $S_1 \setminus V_{m+1}$, one would expect $\mu_i(I^+)$ to measure the number of vertices in $S_1 \cap V_i$, with a contribution of $\delta$ per every such vertex. This suggests an estimation of $|S_1 \cap V_i| \cdot \delta$~ for $\mu_i(I^+)$.
This estimation, however, is not accurate for the following reasons:
\begin{itemize}
  \item The set $V_i$ might include vertices that correspond to imperfect sub-intervals whose measure by $\mu_i$ is smaller than $\delta$. Since there are at most $p(m)$ such vertices in $V_i$, they can cause an error of at most $p(m) \cdot \delta$ in the above estimation.
  \item To make sure that $|V_i|$ is odd, the reduction might add one extra vertex to $V_i$. This might cause an error of at most $\delta$ in the above estimation.
  \item The precise locations $\beta_j$ of the cuts of $I$ might fall inside sub-intervals that correspond to vertices of $V_i$. Since the sub-intervals that correspond to vertices of $V_i$ are disjoint, every such cut can cause an error of at most $\delta$ in the above estimation, and since there are $m+1$ cuts the error here is bounded by $(m+1) \cdot \delta$.
\end{itemize}
We conclude that $\mu_i(I^+)$ differs from the aforementioned estimation $|S_1 \cap V_i| \cdot \delta$~ by not more than $(p(m)+m+2) \cdot \delta$.
Combining~\eqref{eq:|V_i|} and~\eqref{eq:SandV_i}, it can be verified that
\[ \Big | |S_1 \cap V_i| \cdot \delta - \tfrac{1}{2} \Big | \leq \tfrac{\alpha \cdot \eps}{2} + (p(m)+1) \cdot \delta,\]
hence
\begin{eqnarray*}
\Big | \mu_i(I^+) - \tfrac{1}{2} \Big |
&\leq& \Big | \mu_i(I^+) - |S_1 \cap V_i| \cdot \delta \Big | + \Big | |S_1 \cap V_i| \cdot \delta - \tfrac{1}{2} \Big | \\
&\leq& (p(m)+m+2) \cdot \delta + \tfrac{\alpha \cdot \eps}{2} + (p(m)+1) \cdot \delta \\
&=&  \tfrac{\alpha \cdot \eps}{2} + (2p(m)+m+3) \cdot \delta = \tfrac{\eps}{2},
\end{eqnarray*}
where the last equality holds by the definition of $\delta$. This completes the proof.
\end{proof}

We establish the following result.

\begin{theorem}\label{thm:FSplitP'_hard}
There exists a constant $\eps >0$ for which the $\eps\textsf{-}\FSplitP'$ problem is $\PPA$-hard.
\end{theorem}

\begin{proof}
By Theorem~\ref{thm:CH_PPA_hard}, the $\eps\textsf{-}\ConHalv(m,m+1)$ problem is $\PPA$-hard for input density functions that are piecewise constant with at most $3$ blocks, where $\eps >0$ is some constant.
By Theorem~\ref{thm:MainReduction}, for any $\alpha \in [0,1)$, this problem is polynomial-time reducible to the $\frac{\alpha \cdot \eps}{2}\textsf{-}\FSplitP'$ problem, implying the assertion of the theorem.
\end{proof}

\subsection{Hardness of $\FISC$ and $\FSplitC$}

Equipped with Theorem~\ref{thm:FSplitP'_hard}, we are ready to derive the hardness of the $\FISC$ and $\FSplitC$ problems (see Definitions~\ref{def:FISC} and~\ref{def:FSplitC}).

\begin{theorem}\label{thm:FISC_PPAhard}
The $\FISC$ problem is $\PPA$-hard.
\end{theorem}

\begin{proof}
By Theorem~\ref{thm:FSplitP'_hard}, the $\eps\textsf{-}\FSplitP'$ problem is $\PPA$-hard for some $\eps >0$. It thus follows that $\FSplitP'$, with $\eps=0$, is $\PPA$-hard as well. Hence, to prove the theorem, it suffices to show that $\FSplitP'$ is polynomial-time reducible to $\FISC$.

Consider an instance of $\FSplitP'$, that is, a path $G$ on $n$ vertices and a partition $V_1, \ldots ,V_m$ of its vertex set into $m$ sets such that $|V_i|$ is odd for all $i \in [m]$. The reduction simply returns the cycle $G'$, obtained from the path $G$ by connecting its endpoints by an edge, and the same partition $V_1, \ldots ,V_m$ of its vertex set.
For correctness, suppose that we are given a solution to this instance of $\FISC$, i.e., an independent set $S_1$ of $G'$ satisfying $|S_1 \cap V_i| \geq \frac{1}{2} \cdot |V_i|-1$ for all $i \in [m]$.
Since each $|V_i|$ is odd, it can be assumed that $|S_1 \cap V_i| = \frac{1}{2} \cdot (|V_i|-1)$ for all $i \in [m]$ (by removing some vertices from $S_1$ if needed), implying that
\[ |S_1| = \sum_{i=1}^{m}{|S_1 \cap V_i|} = \tfrac{1}{2} \cdot  \sum_{i=1}^{m}{ (|V_i|-1  )} = \frac{n-m}{2}.\]
For every vertex of $S_1$ consider the vertex that follows it in the cycle $G'$ (say, oriented clockwise), and let $S_2$ be the set of vertices that follow those of $S_1$. Since $S_1$ is an independent set in $G'$, we get that $S_2$ is another independent set in $G'$ which is disjoint from $S_1$ and has the same size. We obtain that
\[|S_1 \cup S_2| = |S_1|+|S_2| = 2 \cdot \frac{n-m}{2} = n-m,\]
hence $S_1$ and $S_2$ are two disjoint independent sets of $G'$ covering all but $m$ of its vertices. In particular, $S_1$ and $S_2$ are independent sets in the path $G$, and as such, they form a valid solution to the $\FSplitP'$ instance. This solution can clearly be constructed in polynomial running-time given $S_1$, completing the proof.
\end{proof}

\begin{theorem}\label{thm:FSplitC_eps_hard}
There exists a constant $\eps>0$ for which the $\eps\textsf{-}\FSplitC$ problem is $\PPA$-hard.
\end{theorem}

\begin{proof}
By Theorem~\ref{thm:FSplitP'_hard}, the $\eps\textsf{-}\FSplitP'$ problem is $\PPA$-hard for some constant $\eps >0$.
It thus suffices to show that for every $\eps \geq 0$, the $\eps\textsf{-}\FSplitP'$ problem is polynomial-time reducible to the $\eps\textsf{-}\FSplitC$ problem.

Consider again the reduction that given a path $G$ and a partition $V_1,\ldots,V_m$ of its vertex set into sets of odd sizes returns the cycle $G'$, obtained from the path $G$ by connecting its endpoints by an edge, and the same partition $V_1,\ldots,V_m$. Since the sets of the partition have odd sizes, it follows that the number of vertices and the number of sets in the partition have the same parity, hence the reduction provides an appropriate instance of the $\eps\textsf{-}\FSplitC$ problem.

For correctness, consider a solution to the constructed instance, i.e., two disjoint independent sets $S_1$ and $S_2$ of $G'$ covering all vertices but one from each part $V_i$ such that for each $j \in \{1,2\}$, it holds that $|S_j \cap V_i| \geq (\frac{1}{2}-\eps) \cdot |V_i| -1$ for all $i \in [m]$. We claim that $S_1$ and $S_2$ form a valid solution to the original $\eps\textsf{-}\FSplitP'$ instance. Indeed, an independent set in $G'$ is also an independent set in $G$. In addition, the set $S_1$ satisfies $|S_1 \cap V_i| \in \big [  (\tfrac{1}{2}-\eps  ) \cdot |V_i|-1, (\tfrac{1}{2}+\eps ) \cdot |V_i| \big ]$ for all $i \in [m]$, where the upper bound holds because
\[ |S_1 \cap V_i| = |V_i| - |S_2 \cap V_i|-1 \leq |V_i| - \Big ( (\tfrac{1}{2}-\eps) \cdot |V_i| -1 \Big ) -1 = (\tfrac{1}{2}+\eps) \cdot |V_i|.\]
This completes the proof.
\end{proof}

\begin{remark}\label{remark:0.1}
The $\PPA$-hardness of the $\eps\textsf{-}\ConHalv(m,m+1)$ problem, given by Theorem~\ref{thm:CH_PPA_hard}, was proved in~\cite{DeligkasFHM22} for any constant $\eps < 0.2$.
It thus follows from the proofs of Theorems~\ref{thm:FSplitP'_hard} and~\ref{thm:FSplitC_eps_hard} that $\eps\textsf{-}\FSplitC$ is $\PPA$-hard for any constant $\eps < 0.1$.
\end{remark}

Theorem~\ref{thm:FSplitC_eps_hard} implies that the $\FSplitC$ problem, with $\eps = 0$, is $\PPA$-hard.
The following simple lemma shows that this hardness result can also be derived from the hardness of the $\FISC$ problem, given in Theorem~\ref{thm:FISC_PPAhard}.

\begin{lemma}\label{lem:FISC2FSplitC}
The $\FISC$ problem is polynomial-time reducible to the $\FSplitC$ problem.
\end{lemma}

\begin{proof}
Consider an instance of $\FISC$, that is, a cycle $G$ on $n$ vertices and a partition $V_1, \ldots ,V_m$ of its vertex set into $m$ sets.
If $n$ and $m$ have the same parity then the reduction returns the input as is. Otherwise, there exists some $j \in [m]$ for which the size of $V_j$ is even. In this case, the reduction adds to the cycle $G$ a new vertex located between two arbitrary consecutive vertices and puts it in $V_j$.
Now, the number of vertices and the number of sets in the partition have the same parity, so the reduction can output the obtained cycle and partition.

We turn to prove the correctness of the reduction.
If the given instance of $\FISC$ satisfies that $n$ and $m$ have the same parity, then its solution as an instance of $\FSplitC$ includes two disjoint independent sets that fairly represent the partition, and each of them forms a solution as an instance of $\FISC$ as well.
So suppose that $n$ and $m$ have a different parity, and let $j \in [m]$ denote the index for which the reduction adds a vertex to $V_j$.
Let $u$ denote the added vertex, and define $V'_i = V_i$ for $i \in [m] \setminus \{j\}$ and $V'_j = V_j \cup \{u\}$.
Now, a solution to the constructed instance of $\FSplitC$ includes two disjoint independent sets that fairly represent the partition. Clearly, at least one of the sets does not include both neighbors of $u$. Letting $S'$ denote such a set, it follows that the set $S = S' \setminus \{u\}$ is independent in the original given cycle. For every $i \in [m] \setminus \{j\}$, it holds that $|S \cap V_i| = |S' \cap V'_i| \geq \frac{1}{2}\cdot |V'_i|-1 = \frac{1}{2}\cdot |V_i|-1$. It further holds that
\[|S \cap V_j| \geq |S' \cap V'_j|-1 \geq \tfrac{1}{2}\cdot |V'_j|-2.\]
Since $|V_j|$ is even, it follows that $|V'_j|$ is odd, hence $|S \cap V_j| \geq \frac{1}{2}\cdot |V'_j|-\frac{3}{2} = \frac{1}{2}\cdot |V_j|-1$.
This implies that $S$ is a solution to the original instance of $\FISC$, and we are done.
\end{proof}

\section{The $\SchrijverP$ Problem}\label{sec:Schrijver}

In this section we introduce and study the $\SchrijverP$ problem, a natural analogue of the $\KneserP$ problem defined by Deng et al.~\cite{DengFK17}.

We start with some definitions.
A set $A \subseteq [n]$ is said to be {\em stable} if it does not contain two consecutive elements modulo $n$ (that is, if $i \in A$ then $i+1 \notin A$, and if $n \in A$ then $1 \notin A$). In other words, a stable subset of $[n]$ is an independent set in the cycle on $n$ vertices with the numbering from $1$ to $n$ along the cycle.
For integers $n \geq 2k$, let $\binom{[n]}{k}_{\mathrm{stab}}$ denote the collection of all stable $k$-subsets of $[n]$.
Recall that the Schrijver graph $S(n,k)$ is the graph on the vertex set $\binom{[n]}{k}_{\mathrm{stab}}$, where two sets are adjacent if they are disjoint.
We define the search problem $\SchrijverP$  as follows.

\begin{definition}[Schrijver Graph Problem]\label{def:Schrijver}
In the $\SchrijverP$ problem, the input consists of a Boolean circuit that represents a coloring $$c: \binom{[n]}{k}_{\mathrm{stab}} \rightarrow [n-2k+1]$$ of the Schrijver graph $S(n,k)$ using $n-2k+1$ colors, where $n$ and $k$ are integers satisfying $n \geq 2k$. The goal is to find a monochromatic edge, i.e., two disjoint sets $S_1,S_2 \in \dbinom{[n]}{k}_{\mathrm{stab}}$ such that $c(S_1)=c(S_2)$.
\end{definition}
\noindent
As mentioned earlier, it was proved by Schrijver~\cite{SchrijverKneser78} that the chromatic number of $S(n,k)$ is precisely $n-2k+2$. Therefore, every input to the $\SchrijverP$ problem has a solution.

\subsection{From $\FISC$ to $\SchrijverP$}
The following theorem is used to obtain the hardness result for the $\SchrijverP$ problem. The proof applies an argument of~\cite{AABCKLZ17} (see also~\cite{BCFPS20}).

\begin{theorem}\label{thm:FSplitC2Schrijver}
The $\FISC$ problem is polynomial-time reducible to the $\SchrijverP$ problem.
\end{theorem}

\begin{proof}
Consider an instance of the $\FISC$ problem, namely, a cycle $G$ and a partition $V_1, \ldots ,V_m$ of its vertex set into $m$ sets.
For every $i \in [m]$, let $V'_i$ be the set obtained from $V_i$ by removing one arbitrary vertex if $|V_i|$ is even, and let $V'_i = V_i$ otherwise.
Since the size of every set $V'_i$ is odd, we can write $|V'_i|=2r_i+1$ for an integer $r_i \geq 0$.
Let $G'$ be the cycle obtained from $G$ by removing the vertices that do not belong to the sets $V'_i$ and connecting the remaining vertices according to their order in $G$.
Letting $n$ denote the number of vertices in $G'$, it can be assumed that its vertex set is $[n]$ with the numbering from $1$ to $n$ along the cycle.
Let $k = \sum_{i=1}^{m}{r_i}$, and notice that $n=2k+m$. Define a coloring $c$ of the Schrijver graph $S(n,k)$ as follows. The color $c(S)$ of a vertex $S \in \binom{[n]}{k}_{\mathrm{stab}}$ is defined as the smallest integer $i \in [m]$ for which $|S \cap V'_i| > r_i$ in case that such an $i$ exists, and $m+1$ otherwise. This gives us a coloring of $S(n,k)$ with $n-2k+1$ colors, and thus an instance of the $\SchrijverP$ problem.
It can be seen that a Boolean circuit that computes the coloring $c$ can be constructed in polynomial running-time.

To prove the correctness of the reduction, consider a solution to the constructed $\SchrijverP$ instance, i.e., two disjoint sets $S_1,S_2 \in \binom{[n]}{k}_{\mathrm{stab}}$ with $c(S_1)=c(S_2)$. It is impossible that for some $i \in [m]$ it holds that $|S_1 \cap V'_i| > r_i$ and $|S_2 \cap V'_i| > r_i$, because $S_1$ and $S_2$ are disjoint and $|V'_i|=2r_i+1$. It follows that $c(S_1)=c(S_2)=m+1$, meaning that $|S_1 \cap V'_i| \leq r_i$ and $|S_2 \cap V'_i| \leq r_i$ for all $i \in [m]$. Since $|S_1|=|S_2|=k$, it follows that $|S_1 \cap V'_i| = r_i$ and $|S_2 \cap V'_i| = r_i$ for all $i \in [m]$, hence $S_1$ and $S_2$ are two disjoint independent sets of $G'$ covering all vertices but one from each $V'_i$ and for each $j \in \{1,2\}$, we have $|S_j \cap V'_i| = \frac{1}{2} \cdot (|V'_i|-1) \geq \frac{1}{2} \cdot |V_i|-1$ for all $i \in [m]$. Since $S_1$ and $S_2$ are also independent sets of the original cycle $G$, each of them forms a valid solution to the $\FISC$ instance, completing the proof.
\end{proof}

\subsection{Membership in $\PPA$}

We now show that the $\SchrijverP$ and $\FSplitC$ problems lie in $\PPA$ by reductions to the search problem associated with the Octahedral Tucker lemma.
The reductions follow the proofs of the corresponding mathematical statements by Meunier~\cite{Meunier11} and by Alishahi and Meunier~\cite{AlishahiM17}, and we describe them here essentially for completeness.

We start with some notation (following~\cite[Section~2]{DGMM19}).
The partial order $\preceq$ on the set $\{+,-,0\}$ is defined by $0 \preceq +$ and by $0 \preceq -$, where $+$ and $-$ are incomparable.
The definition is extended to vectors, so that for two vectors $x,y$ in $\{+,-,0\}^n$, we have $x \preceq y$ if for all $i \in [n]$ it holds that $x_i \preceq y_i$ (equivalently, $x_i = y_i$ whenever $x_i \neq 0$).
The Octahedral Tucker lemma, given implicitly in~\cite{Matousek04} and explicitly in~\cite{Ziegler02}, says that for a function $\lambda : \{+,-,0\}^n \setminus \{0\} \rightarrow \{\pm1, \ldots, \pm (n-1)\}$ satisfying $\lambda(-x) = -\lambda(x)$ for all $x$, there exist vectors $x,y$ such that $x \preceq y$ and $\lambda(x) = -\lambda(y)$.
Note that this corresponds to the general Tucker's lemma applied to (the boundary of) the barycentric subdivision of the $n$-cube whose vertex set can be identified with $\{+,-,0\}^n$ (see~\cite{Ziegler02}).
The lemma guarantees the existence of a solution to every input of the following search problem, denoted $\OTucker$.

\begin{definition}[Octahedral Tucker Problem]\label{def:OTucker}
In the $\OTucker$ problem, the input consists of a Boolean circuit that represents a function
$\lambda : \{+,-,0\}^n \setminus \{0\} \rightarrow \{\pm 1, \pm 2, \ldots, \pm (n-1) \}$
satisfying $\lambda(-x) = -\lambda(x)$ for all $x$.
The goal is to find vectors $x,y$ such that $x \preceq y$ and $\lambda(x) = -\lambda(y)$.
\end{definition}
\noindent
The $\OTucker$ problem is known to be $\PPA$-complete~\cite{DengFK17}, where its membership in $\PPA$ essentially follows already from~\cite{Papa94} (see also~\cite[Appendix~A]{DengFK17} and~\cite[Section~3]{AisenbergBB20}).

\begin{proposition}\label{prop:OTucker}
The $\OTucker$ problem lies in $\PPA$.
\end{proposition}

For a vector $x \in \{+,-,0\}^n$, we let $x^+ = \{ i \in [n] \mid x_i=+\}$ and  $x^- = \{ i \in [n] \mid x_i=-\}$. We further let $\alt(x)$ denote the maximum length of an alternating subsequence of $x$, that is, the largest integer $\ell$ for which there exist indices $1 \leq i_1 < i_2 < \cdots < i_\ell \leq n$, such that $x_{i_j} \in \{+,-\}$ for all $j \in [\ell]$ and $x_{i_j} \neq x_{i_{j+1}}$ for all $j \in [\ell-1]$. We clearly have $\alt(x) = \alt(-x)$ for every $x \in \{+,-,0\}^n$.

For a given vector $x \in \{+,-,0\}^n$, we let $A(x)$ denote the vector in $\{+,-,0\}^n$ defined as follows.
Let $I = \{ i\in [n] \mid x_i=0\}$. We first define $A(x)_i = 0$ for every $i \in I$.
Next, consider the restriction $x_{\overline{I}}$ of $x$ to the entries whose indices are in $\overline{I} = [n] \setminus I$, and notice that $x_{\overline{I}}$ can be viewed as a sequence of maximal blocks of $+$'s and of $-$'s. The restriction $A(x)_{\overline{I}}$ of $A(x)$ to the entries of $\overline{I}$ is defined as the vector obtained from $x_{\overline{I}}$ by replacing all the symbols to zeros but the first symbol of each block.
For example, for the vector $x = (+,+,0,+,-,-,0,-,+,0)$ we have $x_{\overline{I}} = (+,+,+,-,-,-,+)$, hence $A(x) = (+,0,0,0,-,0,0,0,+,0)$.

Observe that for every vector $x \in \{+,-,0\}^n$, the vector $A(x)$ satisfies $A(x) \preceq x$ as well as $\alt(x) = \alt(A(x)) = |A(x)^+|+|A(x)^-|$. Further, since $A(x)^+$ and $A(x)^-$ alternate, each of them includes no consecutive numbers.
For an integer $r$ satisfying $r \leq \alt(x)$, we let $A_r(x)$ denote the vector obtained from $A(x)$ by changing its last $\alt(x)-r$ nonzero values to zeros. It clearly holds that $A_r(x) \preceq x$ and that $\alt(A_r(x)) = |A_r(x)^+|+|A_r(x)^-|=r$.
Observe that the quantity $\alt(x)$ and the vectors $A_r(x)$ for $r \leq \alt(x)$ can be computed in polynomial running-time given $x$.

We reduce the $\SchrijverP$ problem to $\OTucker$, applying an argument of~\cite{Meunier11}.

\begin{theorem}\label{thm:Schrijver2OTucker}
The $\SchrijverP$ problem is polynomial-time reducible to the $\OTucker$ problem.
\end{theorem}

\begin{proof}
Consider an instance of the $\SchrijverP$ problem, that is, a Boolean circuit that represents a coloring $c: \binom{[n]}{k}_{\mathrm{stab}} \rightarrow [n-2k+1]$ of the Schrijver graph $S(n,k)$ using $n-2k+1$ colors.
Based on this coloring, we construct an instance of the $\OTucker$ problem given by the function $\lambda : \{+,-,0\}^n \setminus \{0\} \rightarrow \{\pm 1, \pm 2, \ldots, \pm (n-1) \}$ defined as follows. For a given vector $x \in \{+,-,0\}^n \setminus \{0\}$, we consider the following two cases.
\begin{enumerate}
  \item\label{itm:case1} $\alt(x) \leq 2k-1$. \\
  In this case, we define $\lambda(x) = +\alt(x)$ if the first nonzero value of $x$ is $+$, and $\lambda(x) = -\alt(x)$ otherwise.
  \item\label{itm:case2} $\alt(x) \geq 2k$. \\
  In this case, let $z = A_{2k}(x)$ and recall that $\alt(z)=|z^+|+|z^-|=2k$. Observe that $z^+$ and $z^-$ are two vertices of the Schrijver graph $S(n,k)$.
  If $c(z^+) < c(z^-)$ then we define $\lambda(x) = + (c(z^+)+2k-1)$, and if $c(z^-) < c(z^+)$ then we define $\lambda(x) = - (c(z^-)+2k-1)$. Otherwise, we define $\lambda(x)$ to be either $+(n-1)$ or $-(n-1)$, according to whether the first nonzero value of $x$ is $+$ or $-$ respectively.
\end{enumerate}
Since the given coloring $c$ uses the elements of $[n-2k+1]$ as colors, it follows that the function $\lambda$ returns values from $\{\pm 1, \ldots \pm (n-1) \}$.
We further claim that $\lambda(-x) = -\lambda(x)$ for all $x \in \{+,-,0\}^n \setminus \{0\}$. Indeed, this follows from the definition of $\lambda$ combined with the simple fact that for every $x$ and $r \leq \alt(x)$, we have $\alt(x)=\alt(-x)$ and $A_r(x)=-A_r(-x)$.
It is easy to verify that a Boolean circuit that computes the function $\lambda$ can be constructed in polynomial running-time.

We turn to prove the correctness of the reduction. Suppose we are given a solution to the constructed $\OTucker$ instance, i.e., two vectors $x,y \in \{+,-,0\}^n \setminus \{0\}$ with $x \preceq y$ and $\lambda(x) = -\lambda(y)$.
First observe that for a vector $w$ with $\alt(w) \geq 2k$, such that the vector $z = A_{2k}(w)$ satisfies $c(z^+)=c(z^-)$, the sets $z^+$ and $z^-$ are two adjacent vertices in the Schrijver graph $S(n,k)$ and thus they form a monochromatic edge in this graph. Hence, if at least one of the vectors $x$ and $y$ satisfies these conditions, which correspond to the very last sub-case of Case~\ref{itm:case2} in the definition of $\lambda$, then we are done.
Otherwise, by the definition of $\lambda$, either both $\alt(x)$ and $\alt(y)$ are at most $2k-1$ (Case~\ref{itm:case1}) or they are both not (Case~\ref{itm:case2}). However, it is easy to verify that if $x \preceq y$ and $\alt(x) = \alt(y)$ then the first nonzero values of $x$ and $y$ are equal, hence both $\alt(x)$ and $\alt(y)$ must be at least $2k$ (Case~\ref{itm:case2}). Assume without loss of generality that $\lambda(x)$ is positive. By $\lambda(x) = -\lambda(y)$ it follows that $c(A_{2k}(x)^+) = c(A_{2k}(y)^-)$. Using again the fact that $x \preceq y$, it follows that $A_{2k}(x)^+$ and $A_{2k}(y)^-$ are adjacent vertices in the Schrijver graph $S(n,k)$, providing the required monochromatic edge.
\end{proof}

We finally reduce the $\FSplitC$ problem (see Definition~\ref{def:FSplitC}) to $\OTucker$, applying an argument of~\cite{AlishahiM17}.

\begin{theorem}\label{thm:FSplitC2OTucker}
The $\FSplitC$ problem is polynomial-time reducible to the $\OTucker$ problem.
\end{theorem}

\begin{proof}
Consider an instance of the $\FSplitC$ problem, that is, a cycle $G$ on the vertex set $[n]$ and a partition $V_1, \ldots ,V_m$ of $[n]$ into $m$ sets, such that $n$ and $m$ have the same parity. It can be assumed that $n > m$.
We construct an instance of the $\OTucker$ problem given by the function $\lambda : \{+,-,0\}^n \setminus \{0\} \rightarrow \{\pm 1, \pm 2, \ldots, \pm (n-1) \}$ defined as follows. For a given vector $x \in \{+,-,0\}^n \setminus \{0\}$, set
\[J(x) = \Big \{ i \in [m] ~~\Big {|}~~ |x^+ \cap V_i| = |x^- \cap V_i| = \tfrac{|V_i|}{2} \mbox{~ or~} \max(|x^+ \cap V_i|,|x^- \cap V_i|) > \tfrac{|V_i|}{2} \Big \}, \]
and consider the following two cases.
\begin{enumerate}
  \item $J(x) = \emptyset$. \\
  In this case, we define $\lambda(x) = +\alt(x)$ if the first nonzero value of $x$ is $+$, and $\lambda(x) = -\alt(x)$ otherwise.
  Note that by $J(x) = \emptyset$ it follows that $|x^+ \cup x^-| \leq n-m$.
  \item $J(x) \neq \emptyset$. \\
  In this case, let $i$ be the largest element of $J(x)$. If $|x^+ \cap V_i| = |x^- \cap V_i| = \tfrac{|V_i|}{2}$ then we define $\lambda(x) = +(i+n-m-1)$ in the case where the smallest element of $(x^+ \cup x^-) \cap V_i$ is in $x^+$ and $\lambda(x) = -(i+n-m-1)$ otherwise.
  If $\max(|x^+ \cap V_i|,|x^- \cap V_i|) > \tfrac{|V_i|}{2}$ then we define $\lambda(x) = +(i+n-m-1)$ in the case where $|x^+ \cap V_i|> \tfrac{|V_i|}{2}$ and $\lambda(x) = -(i+n-m-1)$ otherwise.
\end{enumerate}
By combining the definition of the function $\lambda$ with the fact that $n >m$, it follows that $\lambda$ returns values from $\{\pm 1, \ldots \pm (n-1) \}$.
We further claim that $\lambda(-x) = -\lambda(x)$ for all $x \in \{+,-,0\}^n \setminus \{0\}$. This indeed follows from the definition of $\lambda$ combined with the fact that for every $x$, we have $\alt(x)=\alt(-x)$ and $J(x)=J(-x)$.
It is easy to verify that a Boolean circuit that computes the function $\lambda$ can be constructed in polynomial running-time.

We turn to prove the correctness of the reduction. Suppose we are given a solution to the constructed $\OTucker$ instance, i.e., two vectors $x,y \in \{+,-,0\}^n \setminus \{0\}$ with $x \preceq y$ and $\lambda(x) = -\lambda(y)$. By the definition of $\lambda$, it is impossible that $J(x) = J(y) = \emptyset$, because if $x \preceq y$ and $\alt(x) = \alt(y)$ then the first nonzero values of $x$ and $y$ are equal. It is also impossible that $J(x)$ and $J(y)$ are both nonempty, because $|\lambda(x)| = |\lambda(y)|$ would imply that the largest element of $J(x)$ is equal to that of $J(y)$, hence by $x \preceq y$, $\lambda(x)$ and $\lambda(y)$ have the same sign.
By $x \preceq y$, we are left with the case where $J(x) = \emptyset$ and $J(y) \neq \emptyset$.
It follows that for some $i \in [m]$, we have
\[ \alt(x) = |\lambda(x)| = |\lambda(y)| = i+n-m-1 \geq n-m.\]
Let $S_1 = x^+$ and $S_2 = x^-$. By $J(x)=\emptyset$, it follows that $|S_1 \cap V_i| + |S_2 \cap V_i| \leq |V_i|-1$ for all $i \in [m]$, and using $\alt(x) \geq n-m$ we get that $|S_1 \cup S_2| = n-m$ and thus $|S_1 \cap V_i| + |S_2 \cap V_i| = |V_i|-1$ for all $i \in [m]$.
This means that $S_1$ and $S_2$ cover all the vertices of $G$ but one from each $V_i$, so by $J(x) = \emptyset$, each of them includes at least $\tfrac{1}{2} \cdot |V_i|-1$ elements of $V_i$. Moreover, the sets $S_1$ and $S_2$ alternate, so since $n-m$ is even, we get that they both form independent sets in the cycle $G$. Hence, $S_1$ and $S_2$ form a valid solution to the given instance of $\FSplitC$, and this solution can be constructed in polynomial running-time given $x$ and $y$.
\end{proof}

\subsection{Putting It All Together}

We finally show that the presented reductions complete the proofs of our results (see Figure~\ref{fig:reductions}).
Indeed, the $\FISC$ problem is $\PPA$-hard by Theorem~\ref{thm:FISC_PPAhard}, and is polynomial-time reducible to the $\SchrijverP$ problem by Theorem~\ref{thm:FSplitC2Schrijver}.
By Theorem~\ref{thm:Schrijver2OTucker}, the latter is efficiently reducible to the $\OTucker$ problem, which by Proposition~\ref{prop:OTucker} lies in $\PPA$. It thus follows that the $\FISC$ and $\SchrijverP$ problems are $\PPA$-complete, as required for Theorems~\ref{thm:FISC} and~\ref{thm:SchrijverP}.
In addition, by Theorem~\ref{thm:FSplitC_eps_hard}, there exists a constant $\eps>0$ for which the $\eps\textsf{-}\FSplitC$ problem is $\PPA$-hard.
The $\eps\textsf{-}\FSplitC$ problem lies in $\PPA$, even for $\eps = 0$, as follows by combining Theorem~\ref{thm:FSplitC2OTucker} with Proposition~\ref{prop:OTucker}. This confirms Theorems~\ref{thm:FSplitC} and~\ref{thm:FSplitC_eps}.

\section*{Acknowledgements}

We are grateful to Aris Filos-Ratsikas and Alexander Golovnev for helpful discussions and to the anonymous referees for their useful suggestions and comments.

\bibliographystyle{abbrv}
\bibliography{fairIScycle}

\begin{thebibliography}{10}

\bibitem{AABCKLZ17}
R.~Aharoni, N.~Alon, E.~Berger, M.~Chudnovsky, D.~Kotlar, M.~Loebl, and R.~Ziv.
\newblock Fair representation by independent sets.
\newblock In M.~Loebl, J.~Ne\v{s}et\v{r}il, and R.~Thomas, editors, {\em A
  Journey Through Discrete Mathematics: A Tribute to Ji{\v{r}}{\'{\i}}
  Matou{\v{s}}ek}, pages 31--58. Springer, 2017.

\bibitem{AisenbergBB20}
J.~Aisenberg, M.~L. Bonet, and S.~Buss.
\newblock 2-{D} {T}ucker is {PPA} complete.
\newblock {\em J. Comput. Syst. Sci.}, 108:92--103, 2020.

\bibitem{AlishahiM17}
M.~Alishahi and F.~Meunier.
\newblock Fair splitting of colored paths.
\newblock {\em Electron. J. Comb.}, 24(3):P3.41:1--8, 2017.

\bibitem{AlonNull99}
N.~Alon.
\newblock Combinatorial {N}ullstellensatz.
\newblock {\em Combinatorics, Probability and Computing}, 8(1--2):7–--29,
  1999.

\bibitem{AlonChallenges02}
N.~Alon.
\newblock Discrete mathematics: methods and challenges.
\newblock In {\em Proc. of the International Congress of Mathematicians
  (ICM'02)}, pages 119--135. Higher Education Press, 2002.

\bibitem{AlonT92}
N.~Alon and M.~Tarsi.
\newblock Colorings and orientations of graphs.
\newblock {\em Combinatorica}, 12(2):125--134, 1992.

\bibitem{AW86}
N.~Alon and D.~B. West.
\newblock The {B}orsuk-{U}lam theorem and bisection of necklaces.
\newblock {\em Proc. Amer. Math. Soc.}, 98(4):623--628, 1986.

\bibitem{BelovsIQSY17}
A.~Belovs, G.~Ivanyos, Y.~Qiao, M.~Santha, and S.~Yang.
\newblock On the polynomial parity argument complexity of the combinatorial
  nullstellensatz.
\newblock In {\em Proc. of the 32nd Computational Complexity Conference
  ({CCC}'17)}, pages 30:1--30:24, 2017.

\bibitem{Berczi017}
K.~B{\'{e}}rczi and Y.~Kobayashi.
\newblock An algorithm for identifying cycle-plus-triangles graphs.
\newblock {\em Discret. Appl. Math.}, 226:10--16, 2017.

\bibitem{BCFPS20}
A.~Black, U.~Cetin, F.~Frick, A.~Pacun, and L.~Setiabrata.
\newblock Fair splittings by independent sets in sparse graphs.
\newblock {\em Israel J. of Math.}, 236:603–--627, 2020.

\bibitem{Borsuk33}
K.~Borsuk.
\newblock Drei {S}{\"{a}}tze {\"{u}}ber die $n$-dimensionale euklidische
  {S}ph{\"{a}}re.
\newblock {\em Fundamenta Mathematicae}, 20(1):177--190, 1933.

\bibitem{ChenD09}
X.~Chen and X.~Deng.
\newblock On the complexity of $2{D}$ discrete fixed point problem.
\newblock {\em Theor. Comput. Sci.}, 410(44):4448--4456, 2009.
\newblock Preliminary version in ICALP'06.

\bibitem{ChenDT09}
X.~Chen, X.~Deng, and S.~Teng.
\newblock Settling the complexity of computing two-player {N}ash equilibria.
\newblock {\em J. {ACM}}, 56(3):14:1--14:57, 2009.
\newblock Preliminary version in FOCS'06.

\bibitem{DaskalakisGP09}
C.~Daskalakis, P.~W. Goldberg, and C.~H. Papadimitriou.
\newblock The complexity of computing a {N}ash equilibrium.
\newblock {\em {SIAM} J. Comput.}, 39(1):195--259, 2009.
\newblock Preliminary version in STOC'06.

\bibitem{DaskalakisP11}
C.~Daskalakis and C.~H. Papadimitriou.
\newblock Continuous local search.
\newblock In {\em Proc. of the 22nd Annual {ACM-SIAM} Symposium on Discrete
  Algorithms ({SODA}'11)}, pages 790--804, 2011.

\bibitem{DGMM19}
J.~A. {De Loera}, X.~Goaoc, F.~Meunier, and N.~H. Mustafa.
\newblock The discrete yet ubiquitous theorems of {C}arath{\'{e}}odory,
  {H}elly, {S}perner, {T}ucker, and {T}verberg.
\newblock {\em Bull. Amer. Math. Soc.}, 56(3):415--511, 2019.

\bibitem{DeligkasFHM22}
A.~Deligkas, J.~Fearnley, A.~Hollender, and T.~Melissourgos.
\newblock Constant inapproximability for {PPA}.
\newblock In {\em Proc. of the 54th Annual {ACM} {SIGACT} Symposium on Theory
  of Computing ({STOC}'22)}, pages 1010--1023, 2022.

\bibitem{DeligkasFM22}
A.~Deligkas, J.~Fearnley, and T.~Melissourgos.
\newblock Pizza sharing is {PPA}-hard.
\newblock In {\em Proc. of the 36th Conference on Artificial Intelligence
  ({AAAI}'22)}, pages 4957--4965, 2022.

\bibitem{DengFK17}
X.~Deng, Z.~Feng, and R.~Kulkarni.
\newblock Octahedral {T}ucker is {PPA}-complete.
\newblock {\em Electronic Colloquium on Computational Complexity {(ECCC)}},
  24:118, 2017.

\bibitem{DengQS12}
X.~Deng, Q.~Qi, and A.~Saberi.
\newblock Algorithmic solutions for envy-free cake cutting.
\newblock {\em Oper. Res.}, 60(6):1461--1476, 2012.

\bibitem{DHH93}
D.~Z. Du, D.~F. Hsu, and F.~K. Hwang.
\newblock The {H}amiltonian property of consecutive-$d$ digraphs.
\newblock {\em Math. and Computer Modelling}, 17(11):61--63, 1993.

\bibitem{ErdosQuestions90}
P.~Erd{\"{o}}s.
\newblock On some of my favourite problems in graph theory and block designs.
\newblock {\em Matematiche}, 45(1):61--73, 1990.

\bibitem{FearnleyGMS19}
J.~Fearnley, S.~Gordon, R.~Mehta, and R.~Savani.
\newblock Unique end of potential line.
\newblock {\em J. Comput. Syst. Sci.}, 114:1--35, 2020.
\newblock Preliminary version in ICALP'19.

\bibitem{FG18}
A.~Filos{-}Ratsikas and P.~W. Goldberg.
\newblock Consensus halving is {PPA}-complete.
\newblock In {\em Proc. of the 50th Annual {ACM} {SIGACT} Symposium on Theory
  of Computing ({STOC}'18)}, pages 51--64, 2018.

\bibitem{FG19}
A.~Filos{-}Ratsikas and P.~W. Goldberg.
\newblock The complexity of splitting necklaces and bisecting ham sandwiches.
\newblock In {\em Proc. of the 51st Annual {ACM} {SIGACT} Symposium on Theory
  of Computing ({STOC}'19)}, pages 638--649, 2019.

\bibitem{FHSZ20}
A.~Filos{-}Ratsikas, A.~Hollender, K.~Sotiraki, and M.~Zampetakis.
\newblock Consensus-halving: Does it ever get easier?
\newblock In {\em Proc. of the 21st {ACM} Conference on Economics and
  Computation ({EC}'20)}, pages 381--399, 2020.

\bibitem{FleischnerS92}
H.~Fleischner and M.~Stiebitz.
\newblock A solution to a colouring problem of {P}.~{E}rd{\"{o}}s.
\newblock {\em Discret. Math.}, 101(1--3):39--48, 1992.

\bibitem{FleischnerS97}
H.~Fleischner and M.~Stiebitz.
\newblock Some remarks on the cycle plus triangles problem.
\newblock In {\em The Mathematics of {P}aul {E}rd{\"{o}}s {II}}, volume~14,
  pages 136--142. Springer, 1997.

\bibitem{GW85}
C.~H. Goldberg and D.~B. West.
\newblock Bisection of circle colorings.
\newblock {\em SIAM J. Alg. Disc. Meth.}, 6(1):93--106, 1985.

\bibitem{GoldbergH19}
P.~W. Goldberg and A.~Hollender.
\newblock The hairy ball problem is {PPAD}-complete.
\newblock {\em J. Comput. Syst. Sci.}, 122:34--62, 2021.
\newblock Preliminary version in ICALP'19.

\bibitem{GoldbergHIMS20}
P.~W. Goldberg, A.~Hollender, A.~Igarashi, P.~Manurangsi, and W.~Suksompong.
\newblock Consensus halving for sets of items.
\newblock In {\em Proc. of the 16th Web and Internet Economics International
  Conference ({WINE}'20)}, pages 384--397, 2020.

\bibitem{Haviv22a}
I.~Haviv.
\newblock A fixed-parameter algorithm for the {K}neser problem.
\newblock In {\em Proc. of the 49th International Colloquium on Automata,
  Languages, and Programming ({ICALP}'22)}, pages 72:1--72:18, 2022.

\bibitem{Haviv22b}
I.~Haviv.
\newblock A fixed-parameter algorithm for the {S}chrijver problem.
\newblock In {\em Proc. of the 17th International Symposium on Parameterized
  and Exact Computation ({IPEC}'22)}, 2022.

\bibitem{JohnsonPY88}
D.~S. Johnson, C.~H. Papadimitriou, and M.~Yannakakis.
\newblock How easy is local search?
\newblock {\em J. Comput. Syst. Sci.}, 37(1):79--100, 1988.
\newblock Preliminary version in FOCS'85.

\bibitem{Kneser55}
M.~Kneser.
\newblock Aufgabe 360.
\newblock {\em Jahresbericht der Deutschen Mathematiker-Vereinigung}, 58(2):27,
  1955.

\bibitem{LovaszKneser}
L.~Lov{\'{a}}sz.
\newblock Kneser's conjecture, chromatic number, and homotopy.
\newblock {\em J. Comb. Theory, Ser. {A}}, 25(3):319--324, 1978.

\bibitem{ManurangsiS19}
P.~Manurangsi and W.~Suksompong.
\newblock Computing a small agreeable set of indivisible items.
\newblock {\em Artif. Intell.}, 268:96--114, 2019.
\newblock Preliminary versions in IJCAI'16 and IJCAI'17.

\bibitem{Matousek04}
J.~Matou{\v{s}}ek.
\newblock A combinatorial proof of {K}neser's conjecture.
\newblock {\em Combinatorica}, 24(1):163--170, 2004.

\bibitem{MatousekBook}
J.~Matou{\v{s}}ek.
\newblock {\em Using the {B}orsuk-{U}lam Theorem: Lectures on Topological
  Methods in Combinatorics and Geometry}.
\newblock Springer Publishing Company, Incorporated, 2007.

\bibitem{MegiddoP91}
N.~Megiddo and C.~H. Papadimitriou.
\newblock On total functions, existence theorems and computational complexity.
\newblock {\em Theor. Comput. Sci.}, 81(2):317--324, 1991.

\bibitem{Meunier11}
F.~Meunier.
\newblock The chromatic number of almost stable {K}neser hypergraphs.
\newblock {\em J. Comb. Theory, Ser. {A}}, 118(6):1820--1828, 2011.

\bibitem{Palvolgyi09}
D.~P{\'{a}}lv{\"{o}}lgyi.
\newblock 2{D}-{T}ucker is {PPAD}-complete.
\newblock In {\em Proc. of the 5th International Workshop on Internet and
  Network Economics ({WINE}'09)}, pages 569--574, 2009.

\bibitem{Papa94}
C.~H. Papadimitriou.
\newblock On the complexity of the parity argument and other inefficient proofs
  of existence.
\newblock {\em J. Comput. Syst. Sci.}, 48(3):498--532, 1994.

\bibitem{Sachs93}
H.~Sachs.
\newblock Elementary proof of the cycle-plus-triangles theorem.
\newblock In {\em Combinatorics, Paul Erd\H{o}s is eighty}, volume~1, pages
  347--359. Bolyai Soc. Math. Stud., 1993.

\bibitem{Schnider21a}
P.~Schnider.
\newblock The complexity of sharing a pizza.
\newblock In {\em Proc. 32nd International Symposium on Algorithms and
  Computation ({ISAAC}'21)}, pages 13:1--13:15, 2021.

\bibitem{SchrijverKneser78}
A.~Schrijver.
\newblock Vertex-critical subgraphs of {K}neser graphs.
\newblock {\em Nieuw Arch.\ Wiskd.}, 26(3):454--461, 1978.

\bibitem{SimmonsS03}
F.~W. Simmons and F.~E. Su.
\newblock Consensus-halving via theorems of {B}orsuk-{U}lam and {T}ucker.
\newblock {\em Math. Soc. Sci.}, 45(1):15--25, 2003.

\bibitem{Ziegler02}
G.~M. Ziegler.
\newblock Generalized {K}neser coloring theorems with combinatorial proofs.
\newblock {\em Invent. math.}, 147(3):671--691, 2002.

\end{thebibliography}

\end{document}